\documentclass[11pt]{article}
\usepackage{sectsty}
\usepackage{graphicx}
\usepackage{soul}
\usepackage{amsmath}
\usepackage{amsthm}
\usepackage[T1]{fontenc}
\usepackage{natbib}
\usepackage{verbatim}
\usepackage[affil-it]{authblk}
\usepackage{algorithm}
\usepackage{algpseudocode}
\usepackage{xcolor}
\usepackage{wrapfig}
\usepackage{titlesec}
\usepackage{enumitem}
\newtheorem{theorem}{Theorem}
\newtheorem{lemma}{Lemma}
\newtheorem{define}{Definition}
\RequirePackage[T1]{fontenc} 
\RequirePackage[tt=false, type1=true]{libertine} 
\RequirePackage[varqu]{zi4} 
\RequirePackage[libertine]{newtxmath}
\titlespacing{\section}{0pt}{4pt}{4pt}
\titlespacing{\subsection}{0pt}{4pt}{4pt}
\titlespacing{\subsubsection}{0pt}{4pt}{4pt}
\setlist[itemize]{leftmargin=*}
\title{\textbf{Algorithms for Multiple Drone-Delivery Scheduling Problem (MDSP)} \vspace{-10pt}}
\author[1]{Sagnik Anupam \footnote{Equal contribution}} 
\author[1]{Nicole Lu \protect{*}}
\author[1]{John Sragow \protect{*}\vspace{-15pt}}
\affil[1]{\small{\textit{Department of Electrical Engineering and Computer Science, MIT, Cambridge, MA, USA}}}
\affil[ ]{\{sanupam, nicolelu, ysragow\}@mit.edu}
\date{\vspace{-50pt}}
\begin{document}
\maketitle
\begin{abstract}
    The Multiple Drone-Delivery Scheduling Problem (MDSP) is a scheduling problem that optimizes the maximum reward earned by a set of $m$ drones executing a sequence of deliveries on a truck delivery route. The current best-known approximation algorithm for the problem is a \(\frac{1}{4}\)-approximation algorithm developed by \cite{jana22}. In this paper, we propose exact and approximation algorithms for the general MDSP, as well as a unit-cost variant. We first propose a greedy algorithm which we show to be a $\frac{1}{3}$-approximation algorithm for the general MDSP problem formulation, provided the number of conflicting intervals is less than the number of drones. We then introduce a unit-cost variant of MDSP and we devise an exact dynamic programming algorithm that runs in polynomial time when the number of drones $m$ can be assumed to be a constant.
\end{abstract}

\section{Introduction}
\noindent
Today, delivery companies can execute last-mile deliveries with the help of drones, which cooperate with a truck on a particular delivery route by flying from the truck, delivering the package to the customer, and returning to the truck. Some formulations of the problem involve both the truck and the drones making deliveries, but the most recent published work in the area has focused on drone-only deliveries and merely use the truck as a carrier for the drones \citep{sorbelli22}. 

In this paper, we explore algorithms for solving the Multiple Drone-Delivery Scheduling Problem (MDSP) variant, a drone-only delivery problem formulation introduced by \cite{sorbelli22}. In their paper, they proved that MDSP is an NP-hard problem \citep[Theorem 3.1]{sorbelli22}. Additionally, they proposed an Integer Linear Programming (ILP) solution for small input sizes, and then also proposed and evaluated three heuristics for use on larger inputs. 

Subsequently, \cite{jana22} developed an $O(n^2)$-runtime greedy $\frac{1}{4}$-approximation algorithm for MDSP. They also introduced the Single-Drone delivery Scheduling Problem (SDSP) variant and an $O(n^2\lfloor\frac{n}{\epsilon}\rfloor)$ Fully Polynomial-Time Approximation Scheme (FPTAS) for it.

In this paper, we seek to modify the greedy algorithm and develop our own $\frac{1}{3}$-approximation algorithm. Additionally, we also examine a unit-cost variant of MDSP and discuss an exact algorithm for it. The paper is structured as follows: first, we provide a comprehensive problem description as well as the ILP formulation for the problem in Section \ref{sec:2}. We then describe \cite{jana22}'s greedy algorithm and present our own modified greedy algorithm in Section \ref{sec:3}. Then, we present the unit-cost variant of MDSP in Section \ref{sec:4}, and provide a dynamic programming algorithm that can generate an exact solution to the variant in polynomial time, provided the number of drones can be taken as a constant. We present our conclusions, as well as directions for future research in Section \ref{sec:5}.

\section{Problem Definition}
\label{sec:2}

The MDSP problem assumes that the delivery company has a pre-defined truck route and pre-calculated launch/rendezvous points for the truck and drones, and based on that route, drone delivery schedules can be created, accounting for the following considerations. \citep{sorbelli22}.

\begin{itemize}
    \item One drone can only serve one customer at a time and must return to the truck to pick up a package for the next customer before embarking on the next delivery.
    \item Delivery destinations have pre-defined launch and rendezvous points and as such, have defined start and end times for the job. Since drones can only serve one customer at a time, jobs whose time intervals overlap cannot be served by the same drone and are said to be in conflict with each other.
    \item Drones have limited battery capacity and cannot recharge during the delivery route (i.e., can only recharge at the depot). Correspondingly, every delivery has an associated energy cost which is not necessarily related to the time required to make the delivery. In the problem variant we study, all drones have the same battery capacity.
    \item Deliveries have specified profits associated with them so deliveries have varying levels of importance.
\end{itemize}

\begin{wrapfigure}{R}{0.5\textwidth}
    \begin{center}
    \includegraphics[width=0.48 \textwidth]{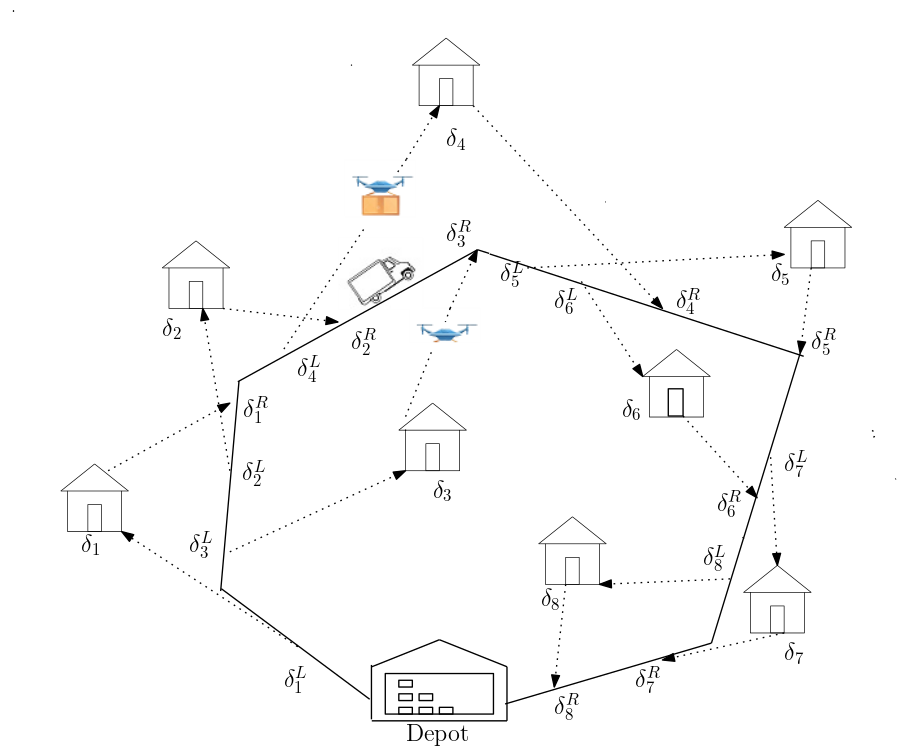}\\
    \caption{Diagram showing the truck delivery route (solid line) and drone delivery routes (dotted lines) and delivery destinations (\(\delta_i\)). Figure from \citep{jana22} and \citep{sorbelli22}}.
    \end{center}
\end{wrapfigure}

The Multiple Drone-Delivery Scheduling Problem (MDSP) is then defined as follows: given the truck's route in the city, plan a schedule for drones such that the total reward is maximized subject to the energy cost and the delivery-conflict constraints. 

Formally, we can describe the MDSP as an Integer Linear Program (ILP) using the following notation as described in \cite{jana22}. Let the truck carry \(m\) identical drones and let there be \(n\) delivery destinations defined as \(D = (\delta_1, \delta_2,\) \(\dots, \delta_n)\). For each destination \(\delta_i\), we construct a pair of two points, the launch point \(\delta^L_i\) and the rendezvous point \(\delta^R_i\), denoted as \((\delta^L_i, \delta^R_i)\). These two points denote the two points on the path of the truck from where the drone delivering the package to the destination \(\delta_i\) launched and landed respectively. 

Now, we describe some of the variables required to specify the scheduling constraints. Let \(B \geq 0\) be the energy budget (battery life) of each of the identical drones, and \(c_j \geq 0\) be the energy cost of executing a delivery to \(\delta_j\). In general, throughout this paper, WLOG we assume that all $c_i \leq B$ i.e. no single job cost exceeds a drone's whole battery capacity (if there exist any such jobs, we remove them from consideration, since no single drone can successfully perform the job).

We now describe the notation used to specify each delivery job. We construct a sequence of times \(t_0=0, t_1, \dots, t_{r+1}\), and assign \(t_0\) to the time the truck leaves the depot and \(t_{r+1}\) to the time the truck returns to the depot. Let times \(t_j^L\) and \(t_j^R\) be the launch and rendezvous times when a drone leaves to execute a delivery job \(\delta_j\) and returns to the truck respectively. 

For each job \(\delta_j\), we specify an interval \(I_j = (t_j^L, t_j^R)\). Two job intervals \(I_a\) and \(I_b\) are said to be in conflict if \(I_a \cap I_b \neq \varnothing\), i.e. the jobs' assigned times overlap; otherwise, the jobs are not in conflict with each other and can potentially be delivered by the same drone. Let \(I = \{I_1, I_2, \dots, I_n\}\) be the set of all the job intervals. We define an \textit{assignment} as a subet \(S_i \subset I\) that denotes the set of jobs successfully completed by drone \(i\). Then, let a \textit{family of assignments} be a set \(S = \{S_1, \dots, S_m\}\) comprising assignments for all $m$ drones. Assignment \( S_i \subseteq I\) is said to be \textit{compatible} if no two jobs in \(S_i\) are in conflict. \(S_i\) is said to be \textit{feasible} if it satisfies the energy constraint i.e. the cost function $\mathcal{W}$ satisfies \(\mathcal{W}(S_i) = \sum_{I_j \in S_i} c_j \leq B\).

Let \(p_j \geq 0\) be the reward for executing a delivery \(\delta_j\). For the rest of the paper, we define the total reward of an assignment \(S_i\) as \(\mathcal{P}(S_i) = \sum_{I_j \in S_i} p_j\). We also define the total reward of a family of assignments \(S\) as \(\mathcal{P}(S) = \sum_{1 \leq i \leq m } \mathcal{P}(S_i)\). Thus, to find the optimal family of assignments to solve MDSP, we have to find \(S^* = \text{arg max } \mathcal{P}(S)\) over all possible families of compatible and feasible assignments $S$. We hence have the ILP as formulated in \cite{sorbelli22} and \cite{jana22} presented in (\ref{eq:1}-\ref{eq:5}), where the decision variable \(x_{ij} \in \{0, 1\}\) is used to denote if the delivery \(j\) was executed by drone \(i\). The first constraint enforces the energy constraint, while the second ensures that each delivery is executed only once, and the third ensures that no two deliveries made by the same drone are overlapping.

\begin{wrapfigure}{r}{.5\textwidth}
\begin{align}
\label{eq:1}
& \text{max }\sum_{i=1}^m\sum_{j=1}^np_jx_{ij} \text{, subject to:} \\
& \sum_{j=1}^nc_jx_{ij} \leq B\;\;\;\;\forall\;i \in \mathcal{M} \\
& \sum_{i=1}^mx_{ij} \leq 1\;\;\;\;\forall\;j \in \mathcal{N}\;\; \\
& x_{ij} + x_{ik} \leq 1\;\forall\;i \in \mathcal{M};\;\forall\;j, k \in \mathcal{N}\;\text{s.t}\;I_j \cap I_k = \varnothing \\
& x_{ij}\in\{0,1\}\;\;\;\;\forall\;i\leq n,\;j \leq m
\label{eq:5}
\end{align}
\end{wrapfigure}

The ILP formulation is helpful for understanding the structure and constraints of the problem but it is difficult to apply the formulation to solve problems with large input sizes. Thus, \cite{sorbelli22} proposed three heuristics to discover solutions to the problem.

The first of their three proposed heuristics, the \emph{MR-S} algorithm, is an $O(n \text{ log } n)$ time-complexity greedy algorithm that sorts jobs by profit-cost density and assigns non-conflicting jobs to a single drone in descending profit-cost density order. 

The second heuristic, the \emph{MC-M} algorithm, calculates cliques formed by interval conflicts in the job list and sequentially performs partitioning on the maximum clique induced by the remaining un-scheduled jobs. This is repeated until all jobs are scheduled. The runtime of this heuristic is $O(m(n \text{ log } n + h(n)))$, where $h(n)$ is the runtime of the optimal partition generation algorithm. 

Their final heuristic, the \emph{MR-M} algorithm, has $O(m(n\text{ log } n))$ runtime. Its uses their single drone \emph{MR-S} algorithm and schedules multiple drones sequentially by running the greedy single-drone algorithm on one drone at a time. All three heuristics have $O(n)$ space complexity \citep{sorbelli22}.

However, more recent work on this problem has focused on developing various approximation algorithms for the problem \citep{jana22}. Our work builds in that same direction, as we wish to develop algorithms that provide better approximation factors than the existing algorithms. The performance guarantees for these approximation algorithms rely on the maximum number of overlapping jobs at any point in time, $\Delta$, and in accordance with previous work done on this problem, we focus on the cases where $\Delta \leq m$ i.e. the maximum number of overlapping jobs at any point in time is at most $m$. In the next section, we thus describe \cite{jana22}'s $\frac{1}{4}$-approximation algorithm, as well as present our own $\frac{1}{3}$-approximation algorithm.

\section{Approximation Algorithms for MDSP}
\label{sec:3}

\subsection{\cite{jana22}'s Greedy \(\frac{1}{4}\)-approximation Algorithm}

\cite{jana22} proposed a greedy algorithm that yields a $\frac{1}{4}$-approximation, provided that the number of drones is at least the maximum number of conflicts any job. We describe the algorithm in this section:

\subsubsection{Definitions}
\begin{define}[Densities]
    The density of a job \(j\) is the ratio of its reward to its cost i.e. \(d_j = \frac{p_j}{c_j}\).
\end{define}
\begin{define}[Critical assignment]
    A critical (or overfull) assignment is  an assignment of jobs to drones that violates the battery constraint by 1 job. I.e., if the least-profit-dense job was removed, then the assignment would obey the drone's battery constraint.

    We define the $i$-th drone's assignment $S_i = \{I_1, I_2, \ldots I_{k}\} \subseteq I$ with densities, $d_1, d_2, \ldots, d_{k}$ as critical if the total battery cost $\mathcal{W}(S_i)$ of $S_i$ violates the drone's battery constraint, but the subset $S'_i = \{I_1, I_2, \ldots I_{k-1}\} $ obeys the battery constraint.
    
\end{define}

\subsubsection{Algorithm Description}
The greedy algorithm has the following high-level logic:
\begin{itemize}
    \item Find the maximum number of overlapping delivery intervals $(\Delta)$ and add that number of drones to our set of available drones.
    \item Greedily assign jobs to drones (in descending order of profit-cost density) until their batteries are overfull (note that the algorithm specifically looks to overfill drones and keeps them in consideration for new jobs even when their battery is exactly full).
    \item For the overfull drones, go back and remove low profit-density job(s) such that the battery constraint is not violated.
    \item Return the $m$ drones whose schedules yield the most profit.
\end{itemize}

\noindent
More specifically, the algorithm works as described in Algorithm \ref{alg:1}.

\begin{algorithm}[h]
\caption{GreedyAlgorithmMDSP}
\begin{algorithmic}[1]
\State \text{Construct interval graph for delivery intervals with max degree $\Delta$}
\State \text{Sort deliveries by profit-cost densities $(d)$ WLOG $d_1 \geq d_2 \geq 
\ldots d_n$ }
\State \text{Let $M = \{1, \dots, m+\Delta\}$ be the set of available drones and $M' = \varnothing$ be the set of critically-assigned drones}
\State \text{Let $W_i = \mathcal{W}(S_i)$ and $P_i = \mathcal{P}(S_i)$ be the total respective cost and reward of drone $i$'s current schedule $S_i$}
\For { $j \gets 1, n$}
\If {$ |M'| < m$ \textit{(fewer than $m$ drones have been critically assigned)}}
\State \text{Find $i \in M : I_k \cap I_j = \varnothing \; \forall \; I_k \in S_i$} \textit{(find drone whose schedule doesn't conflict with current job)}
\State {$W_i = W_i + c_j; S_i = S_i \cup \{I_j\}; P_i = P_i + p_j$ \textit{(update schedule)}}
\If {$W_i + c_j > B$ \textit{\; (battery constraint is violated by adding job)}}
\State {$M = M \setminus \{i\}; M' = M' \cup \{i\}; L_i = \{j\}$ \textit{(mark drone as critically-assigned, note last added job)}}
\EndIf
\Else
\State {\text{break} \textit{($m$ critically-assigned drones) }}
\EndIf
\EndFor
\For {\text{each $i \in M'$} \textit{(for all critically-assigned drones $i$)}}
\State {\text{Let $L_i$ denote the least-dense job in $S_i$ (which was added last)}} 
\State{\text{Let $\mathcal{P}(L_i)$ and $\mathcal{W}(L_i)$ be $L_i$'s reward and cost respectively}}
\If {$P_i - \mathcal{P}(L_i) \geq \mathcal{P}(L_i)$ \textit{(reward of $L_i$ is less than or equal to the reward of the rest of $S_i$)}}
\State {$W_i = W_i - \mathcal{W}(L_i); P_i = P_i - \mathcal{P}(L_i); S_i = S_i \setminus \{L_i\}$ \textit{(remove $L_i$ from $i$ to make $W_i \leq B$)}} 
\Else {\textit{ ($L_i$ has more reward than rest of $S_i$)}}
\State {$W_i = \mathcal{W}(L_i); P_i = \mathcal{P}(L_i); S_i = \{L_i\}$ \textit{($S_i$ now only contains $L_i$)}}
\EndIf
\EndFor
\State \text{Return $m$ most profitable drones in $M \cup M'$ with their respective schedules}
\end{algorithmic}
\label{alg:1}
\end{algorithm}

\subsubsection{Runtime and feasibility}

The algorithm runs in $O(n^2)$ time \citep{jana22} and is always guaranteed to return a feasible solution (if one exists) since we know there are only two ways a potential job scheduling can be infeasible:
\begin{itemize}
    \item \emph{Overlapping time intervals:} A drone's schedule can include jobs whose time intervals overlap each other and thus cannot be feasibly completed by the same drone
    \item \emph{Battery constraint:} A drone's schedule can include jobs that together require more battery capacity than the drone can support
\end{itemize}

\noindent
Now, since the algorithm goes through overfull drones and deletes the last-added job which made the drone overfull, no drone has an overfull schedule in the algorithm's assignment, which enforces the battery capacity constraint.

Since the greedy-assignment portion of the algorithm looks for a drone whose schedule does not have a scheduling conflict with the given job, we know that jobs are never assigned to drones in such a way that scheduled jobs have timing conflicts. Since a job can have at most $\Delta$ conflicts, there is always at least one drone available whose schedule does not conflict with the given job, so the algorithm will never generate a schedule with conflicts while assigning jobs. 

It's worth noting that, because this algorithm hinges on creating at least $m$ critically assigned drone schedules, we rely on the $\Delta$ extra drones we've added to be able to handle conflicts and critically assign drones (even though $m \geq \Delta$). Consider a situation where we only have $m$ drones. If we are still assigning jobs to drones, then we must have fewer than $m$ critically assigned drones. For any job $j$ we are trying to assign, we are guaranteed to have at least $m - \Delta$ drones that don't already hold conflicting jobs. However, there is no guarantee that those $m - \Delta$ drones are not already critically assigned. If we instead use $m + \Delta$ drones, then we are guaranteed to have $m + \Delta - \Delta = m$ drones that do not hold an overlapping job. Since we are still assigning jobs, we must have fewer than $m$ critically assigned drones, which implies we must have at least one suitable drone for us to place job $j$.

\subsubsection{Approximation Factor}

For $m$ drones and $\Delta$ maximum overlapping jobs, the algorithm proposed in \cite{jana22} yields a $\frac{m}{2(m + \Delta)}$ approximation. Let $\mathcal{S} = \{S_1, \dots, S_{m+\Delta}\}$ denote the family of assignments generated for $m+\Delta$ drones by steps 1-15 of the algorithm (prior to correction), including either $m$ overfull drones or all $n$ jobs assigned to $m+\Delta$ drones. Then, let $\mathcal{S}' = \{S_1', \dots, S_{m+\Delta}'\}$ denote the "corrected" family of assignments generated for $m+\Delta$ drones by the entire algorithm (by steps 16-25). Finally, we define $\mathcal{S}_A = \{S_{A_1}, \dots, S_{A_m}\}$ as the family of assignments comprising the $m$ highest-reward drone schedules in $\mathcal{S}'$. Now, we use the following theorems to prove the correctness of this approximation factor: 

\begin{theorem}
Given a drone $i$ with an overfull assignment $S_i$, denote the "corrected" assignment generated via steps 16-25 of the algorithm as $S_i'$. Then, we have $\mathcal{P}(S'_i) \geq \frac{\mathcal{P}(S_i)}{2}.$
\end{theorem}

\noindent
\begin{proof}
If the least profit-dense job has profit greater than the rest of the jobs combined, the algorithm will keep that job and delete the rest. Otherwise, it will delete the one least profit-dense job. Given that the smaller of two choices is removed, the algorithm removes at most half of the drone's profit. \end{proof}

\begin{theorem}
\label{thm:2}
$\mathcal{P}(\mathcal{S}) \geq OPT$, i.e. the profit of the greedily assigned family of assignments pre-correction is at least OPT.
\end{theorem}


\noindent
\begin{proof}
We present a proof of this theorem adapted from the one provided in \cite{jana22}. We split the statement into two cases: one when $\mathcal{S}$ has $m$ critically-assigned/overfull drones, and one when $\mathcal{S}$ assigns all $n$ jobs to $m+\Delta$ drones before $m$ drones are overfull. 

For the case where $\mathcal{S}$ has $m$ critically-assigned/overfull drones, let us denote $\mathcal{A}$ as the set of all jobs in $\mathcal{S}$, and $\mathcal{B}$ as the set of all jobs in the optimal job assignment. We also define a density function of a set of jobs as follows: $d(X) = \frac{\mathcal{P}(X)}{\mathcal{W}(X)}$. We shall use the manner of construction of the two families of assignments to show a contradiction in the properties of the cost function of the two families of assignments. 

Suppose for the sake of contradiction that $\mathcal{P}(\mathcal{S}) < OPT$, which implies that $\mathcal{P}(\mathcal{A}) < \mathcal{P}(\mathcal{B})$, and hence indicates that $\mathcal{P}(\mathcal{A} \setminus \mathcal{B}) < \mathcal{P}(\mathcal{B} \setminus \mathcal{A})$. Since the greedy algorithm assigns jobs greedily based on profit-cost density, $\mathcal{S}$ already contains the densest jobs that could be used to overfill $m$ drones of $B$ battery capacity each. This statement holds because $\mathcal{S}$ includes $m$ overfull drones as well as all jobs denser than the jobs assigned to the overfull drones. Hence, $d(\mathcal{A}) \geq d(\mathcal{B})$, and thus $d(\mathcal{A} \setminus \mathcal{B}) \geq d(\mathcal{B} \setminus \mathcal{A})$. Since density is profit divided by cost, $d(\mathcal{A} \setminus \mathcal{B}) \geq d(\mathcal{B} \setminus \mathcal{A})$ and $\mathcal{P}(\mathcal{A} \setminus \mathcal{B}) < \mathcal{P}(\mathcal{B} \setminus \mathcal{A})$ together imply  $\mathcal{W}(\mathcal{A} \setminus \mathcal{B}) < \mathcal{W}(\mathcal{B} \setminus \mathcal{A}) \implies \mathcal{W}(\mathcal{A}) < \mathcal{W}(\mathcal{B})$. 

We note that the drones in $\mathcal{S}$ are critically assigned, and thus $\mathcal{W}(\mathcal{S}) > mB$. Similarly, the drones in $\mathcal{T}$ are feasibly assigned, and thus $\mathcal{W}(\mathcal{T}) \leq mB$. Thus, $\mathcal{W}(\mathcal{S}) < \mathcal{W}(\mathcal{T})$ is a contradiction, and hence $\mathcal{P}(\mathcal{S}) \geq OPT$ must hold.

We can prove the second case, where all the jobs are distributed between $m+\Delta$ drones by seeing that when the algorithm terminates after assigning all jobs, the total profit of the assigned jobs (across all of the $m + \Delta$ drones) is equal to the total profit of all jobs. Since the maximum profit is at most the profit of all jobs, so $\mathcal{P}(\mathcal{S}) = \mathcal{P}(\mathcal{T})$. Hence, for both cases, we have $\mathcal{P}(\mathcal{S}) \geq OPT$.
\end{proof}

\begin{theorem}
    The greedy algorithm is $\frac{1}{4}$-approximate.
    \label{thm:3}
\end{theorem}
\noindent
\begin{proof}
Since $\mathcal{S}_A$ comprises the $m$ highest-reward drones from $\mathcal{S}'$, the total reward of $\mathcal{S}_A$, denoted by $\mathcal{P}(\mathcal{S}_A)$, is lower-bounded by the average profit of the $\mathcal{S}'$ family of assignments for $m$ drones. Assuming that we had $t \leq m$ overfull drones that were "corrected" by the algorithm in steps 16-25, we thus have the inequality in (\ref{eq:2}) as per \cite{jana22}:

\begin{equation} \mathcal{P}(S_A) \geq \frac{m}{m+\Delta} \left(\sum_{i=1}^t \mathcal{P}(S_i') + \sum_{i=t+1}^{m+\Delta} \mathcal{P}(S_i)\right) \geq \frac{m}{2(m+\Delta)} \left(\sum_{i=1}^{m+\Delta} \mathcal{P}(S_i)\right) \geq \frac{m\cdot OPT}{2(m+\Delta)}
\label{eq:2}
\end{equation}

\noindent
This provides a $\frac{1}{4}$-approximation algorithm if $\Delta \leq m$.
\end{proof}

\subsection{A Modified $\frac{1}{3}$-approximate Greedy Algorithm for MDSP}

Our modification to \cite{jana22}'s approximation algorithm involves changing the number of drones we have initially from $m+\Delta$ to $2m + \Delta$ to improve the approximation factor. This is because a key driver of the approximation factor of \cite{jana22}'s algorithm is the step to remove from each $S_i$ in $\mathcal{S} = \{S_1, S_2, \dots, S_{m+\Delta}\}$ either the last-added, least-dense job ($L_i$) or the rest of the jobs ($S_i \setminus \{L_i\}$). Thus, after the "correction" step, we have:

\begin{equation}
\mathcal{P}(\mathcal{S}) = \sum_{i=1}^{m+\Delta} \mathcal{P}(S_i) \geq \mathcal{P}(\mathcal{S}') = \sum_{i=1}^{m+\Delta} \mathcal{P}(S'_i) \geq \frac{\mathcal{P}(\mathcal{S})}{2}
\end{equation}

We propose a modification to \cite{jana22}: instead of deleting these jobs during the correction phase, we store them in $m$ additional, separate drones that are not used in steps 1-15 of the algorithm and are only brought into play for steps 16-25. Since we have at most $m$ critically assigned drones, we originally deleted either of the following two quantities:

\begin{itemize}
    \item the least-dense job $L_i$, where we know $\mathcal{W}(L_i) \leq B$
    \item the rest of the schedule $S_i \setminus \{L_i\}$, where we know $\mathcal{W}(S_i \setminus \{L_i\}) \leq B$ as $L_i$ was the job that made drone $i$ critically-assigned.
\end{itemize}

\noindent
Hence, in the correction phase of our algorithm, we will have to relocate at most $m$ quantities, each having cost less than or equal to $B$. This can easily be done using our $m$ new, empty drones, and so we will never have to discard any jobs. Not having to discard jobs will lead us to a tighter approximation bound because we can avoid the conservative assumption that the discarded jobs are at worst half of the pre-corrected drone profit. We present the algorithm formally as Algorithm \ref{alg:2}.

\begin{algorithm}[h]
    \caption{ModifiedGreedyAlgorithmMDSP}
    \begin{algorithmic}[1]
    \State \text{Construct interval graph for delivery intervals with max degree $\Delta$}
    \State \text{Sort deliveries by profit-cost densities $(d)$ WLOG $d_1 \geq d_2 \geq 
    \ldots d_n$ }
    \State \text{Let $M = \{1, \dots, 2m+\Delta\}$ be the set of available drones and $M' = \varnothing$ be the set of critically-assigned drones}
    \State \text{Let $W_i = \mathcal{W}(S_i)$ and $P_i = \mathcal{P}(S_i)$ be the total respective cost and reward of drone $i$'s current schedule $S_i$}
    \For { $j \gets 1, n$}
    \If {$ |M'| < m$ \textit{(fewer than $m$ drones have been critically assigned)}}
    \State \text{Find $i \in M : I_k \cap I_j = \varnothing \; \forall \; I_k \in S_i$} \textit{(find drone whose schedule doesn't conflict with current job)}
    \State {$W_i = W_i + c_j; S_i = S_i \cup \{I_j\}; P_i = P_i + p_j$ \textit{(update schedule)}}
    \If {$W_i + c_j > B$ \textit{\; (battery constraint is violated by adding job)}}
    \State {$M = M \setminus \{i\}; M' = M' \cup \{i\}; L_i = \{j\}$ \textit{(mark drone as critically-assigned, note last added job)}}
    \EndIf
    \Else
    \State {\text{break} \textit{($m$ critically-assigned drones) }}
    \EndIf
    \EndFor
    \For {\text{each $i = 1, 2, \dots, |M'|$} \textit{(for all critically-assigned drones $i$ labelled by the corresponding integer)}}
    \State {\text{Let $L_i$ denote the least-dense job in $S_i$ (which was added last)}} 
    \State{\text{Let $\mathcal{P}(L_i)$ and $\mathcal{W}(L_i)$ be $L_i$'s reward and cost respectively}}
    \If {$P_i - \mathcal{P}(L_i) \geq \mathcal{P}(L_i)$ \textit{(reward of $L_i$ is less than or equal to the reward of the rest of $S_i$)}}
    \State {$W_i = W_i - \mathcal{W}(L_i); P_i = P_i - \mathcal{P}(L_i); S_i = S_i \setminus \{L_i\}; S_{m+\Delta+i} = \{L_i\}$ \textit{(remove and reassign $L_i$ from $i$ to make $W_i \leq B$)}} 
    \Else {\textit{ ($L_i$ has more reward than rest of $S_i$)}}
    \State {$W_i = \mathcal{W}(L_i); P_i = \mathcal{P}(L_i); S_{m+\Delta+i} = S_i\setminus\{L_i\}; S_i = \{L_i\}$ \textit{($S_i$ now only contains $L_i$)}}
    \EndIf
    \EndFor
    \State \text{Return $m$ most profitable drones in $M \cup M'$ with their respective schedules}
    \end{algorithmic}
    \label{alg:2}
    \end{algorithm}

Since we simply reassign the discarded quantities to a new drone in constant time instead of deleting them, the runtime of the algorithm stays $O(n^2)$ as we only increase the runtime by $O(|M'|)<O(n^2)$ \citep{jana22}. Using the same notation as earlier, let $\mathcal{S} = \{S_1, \dots, S_{m+\Delta}\}$ denote the family of assignments generated for $m+\Delta$ drones by our modified algorithm (prior to correction), including either $m$ overfull drones or all $n$ jobs assigned to $m+\Delta$ drones. Then, let $\mathcal{S}^{*} = \{S_1^{*}, \dots, S_{2m+\Delta}^{*}\}$ denote the "corrected" family of assignments generated for $2m+\Delta$ drones by our modified algorithm. Finally, we define $\mathcal{S}^*_A = \{S^*_{A_1}, \dots, S^*_{A_m}\}$ as the family of assignments comprising the $m$ highest-reward drone schedules in $\mathcal{S}^{*}$. We now show that the total reward from our family of assignments for $2m+\Delta$ drones is greater than the reward from the optimal solution in Theorem \ref{thm:4}, and then use the theorem to prove that our algorithm is a $\frac{1}{3}$-approximation algorithm.

\begin{theorem}
\label{thm:4}
Given our new family of assignments $\mathcal{S}^*$, $\mathcal{P}(\mathcal{S}^*) \geq OPT.$
\end{theorem}

\noindent
\begin{proof}
We know that the total cost and the total reward for the $2m+\Delta$ drones in $\mathcal{S}^*$ is exactly the same as the cost and the reward of $\mathcal{S}$ (the pre-correction family of assignments of \cite{jana22}'s algorithm), since none of the jobs in $\mathcal{S}$ were ever removed while constructing $\mathcal{S}^*$, merely moved to $m$ additional, different drones. Thus, from Theorem \ref{thm:2}, we have $\mathcal{P}(\mathcal{S}^*) = \mathcal{P}(\mathcal{S}) \geq OPT$.
\end{proof}

\begin{theorem}
\label{thm:5}
Our modified greedy algorithm is $\frac{1}{3}$-approximate.
\end{theorem}

\noindent
\begin{proof}
We know that $\mathcal{P}(\mathcal{S}^*) = \mathcal{P}(\mathcal{S})$, as none of the jobs in $\mathcal{S}$ were ever removed while constructing $\mathcal{S}^*$. We note that this property does not hold for $\mathcal{P}(\mathcal{S}')$, i.e. $\mathcal{P}(\mathcal{S}') \neq \mathcal{P}(\mathcal{S})$ as \cite{jana22}'s algorithm deletes at least one job from every critically-assigned drone, which leads to $\mathcal{P}(\mathcal{S}) \geq \mathcal{P}(\mathcal{S}') \geq \frac{\mathcal{P}(\mathcal{S})}{2}$. Now, we use Theorem \ref{thm:4} and bound $\mathcal{P}(\mathcal{S}_A^*)$ by average profits to obtain the inequality in (\ref{eq:3}):

\begin{equation}
\mathcal{P}(\mathcal{S}_A^*) \geq \frac{m \cdot \mathcal{P}(\mathcal{S}^*)}{2m+\Delta} = \frac{m\cdot \mathcal{P}(\mathcal{S})}{2m+\Delta} \geq \frac{m\cdot OPT}{2m+\Delta}.
\label{eq:3}
\end{equation}

\noindent
Thus, our modified greedy algorithm is a $\frac{1}{3}$-approximation algorithm if $\Delta \leq m$.
\end{proof}

\section{Polynomial-Time Exact Algorithm for Unit Cost MDSP (Constant-$m$)}
\label{sec:4}

One significant challenge of the MSDP problem is the variability of delivery costs (e.g., no bounds, no guaranteed connection to delivery time). In real life, many delivery scenarios would have more predictable delivery costs. For example, in many urban or suburban areas, the front door of a building is within some reliable distance of the street and is at roughly the same elevation as the street. In these situations, the battery cost of each delivery would be approximately equal. We can model these scenarios with unit-cost deliveries. If we assume that the cost of every job is 1, 
we can solve the problem in polynomial time for constant $m$. 

In this section, we show that if all jobs have unit cost, we can solve MDSP in polynomial time assuming that $m$ is a constant. The algorithm we develop to prove this statement will use dynamic programming.  We first describe the subproblems for this algorithm and then define the Maximum Profit Fitting Family of Assignments (MPFFA), the maximum-profit family of assignments that solves our subproblem. Then, we state and prove a lemma that proves that we can construct certain families of assignments that are MPFFAs from other MPFFAs, allowing us to connect our subproblems. After proving the lemma, we will use it to construct MPFFAs for all our subproblems and show that the optimal family of assignments is the maximum-profit MPFFA across all subproblems. Finally, we prove that the runtime of our algorithm is polynomial in the number of jobs.

First, assume that $B > n$.  In this case, because the cost of each job is 1, and there are $n$ jobs total, it will never be possible to hit the cost constraint.  Now, assume that $B \leq n$. In this case, because each job has unit cost, the only possible total costs on each drone are integers between 0 and $n$.  Therefore, no matter what the value of $B$ is, the only possible values of the total cost assigned to a given drone is an integer between $0$ and $\min(n,B)$.\par 
Using that knowledge, we can define our subproblems in terms of the total cost taken by each drone. Let $S_i$ be the assignment to drone $i$. Let $t_i$ be the time at which the latest job assigned to drone $i$ finishes, i.e. $t_i$ is the end time $t_j^R$ of a job $j$ such that $t_j^R > t_k^R$ for any other job $k$ in drone $i$'s schedule $S_i$. Let $\mathcal{W}(S_i)$ be the total cost of the jobs in assignment $S_i$.  As such, we can define each subproblem $Q$ as:  
\begin{equation}
Q = [t_1... t_m,\mathcal{W}(S_1)...\mathcal{W}(S_m)]
\end{equation}
In other words, the number of subproblems we have for this dynamic programming problem is every single combination of every latest job time on every drone and every total cost on every drone.\par 

Next, we need to define some terms:
\begin{define}[Fitting]
    A family of assignments $A$ is a \textbf{\emph{fitting}} family of assignments for a subproblem $Q = [t_1... t_m,\mathcal{W}(S_1)...\mathcal{W}(S_m)]$ if for each drone $i$ in $A$, the latest endtime of any job assigned to that drone is $t_i$ and the total cost of jobs assigned to drone $i$ is $\mathcal{W}(S_i)$.
\end{define}
\begin{define}[Maximum Profit Fitting Family of Assignments (MPFFA)]
    A \textbf{\emph{maximum profit fitting family of assignments (MPFFA)}} for a subproblem $Q$ is a fitting family of assignments $A$ for $Q$ such that the profit $\mathcal{P}(A)$ of $A$ is greater than or equal to the profit of any other fitting family of assignments for $Q$.
\end{define}
\begin{define}[Top]
    Let $\;Q = [t_1... t_m,\mathcal{W}(S_1)...\mathcal{W}(S_m)]$ be a subproblem.  Let $L_Q = \{t_1... t_m\}$ be the set containing the latest endtime of every drone in $Q$.  Let $t_Q = \max_{t_i\in L_Q}(t_i)$ be the latest endtime in $L_Q$.  We can thus define the \textbf{top} of $Q$, $T_Q$, as $T_Q =  \{i\leq n\;|\;t_i = t_Q\}$, or the set of drones whose latest endtime is equal to $t_Q$.
\end{define}
Using these terms, we can now state the following lemma:
\begin{lemma}
    Let $Q = [t_1... t_m,\mathcal{W}(S_1)...\mathcal{W}(S_M)]$ be a subproblem, and let $A$ be an MPFFA for $Q$.  Let $J$ be the set of jobs $j_i$ assigned in $A$ such that the endtime $t_{j_i}^R$ of $j_i$ is $t_i$, for every $i\in T_Q$, where $T_Q$ is the top of $Q$.  Remove every job in $J$ from $A$ to get a new job assignment $A'$.  Let $Q'$ be the subproblem for which $A'$ is a fitting family of assignments.  Then, $A'$ is an MPFFA for $Q'$.
\end{lemma}
In other words, if we remove the last job from every drone in the top of an MPFFA, then the resulting family of assignments is also an MPFFA.  Now, we will prove this lemma.
\begin{proof}
Assume by contradiction that $A'$ is not an MPFFA for $Q'$.  That means there exists some fitting family of assignments $D'$ for $Q'$ such that $\mathcal{P}(D') > \mathcal{P}(A')$.  \par 
First, we will show that any job in $J$ cannot have been assigned in $D'$.  By how $J$ is defined, every job $j\in J$ ends at time $t_Q$, and $t_Q$ is later than the endpoint of every job in $A$.  Thus, by definition, there cannot exist a job in $A$ that ends later than $t_Q$, which means that any job in $A$ which ends at $t_Q$ must be in $J$.  Thus, when every job in $J$ is removed to get $A'$, it means that every job ending at $t_Q$ was removed, and since no jobs ending after $t_Q$ can be in $A$ or $A'$, we get that the latest ending time of any job in $A'$ must be some $t_{Q'} < t_Q$.  This means that every drone in a fitting family of assignments for $Q'$ must have latest endpoints $t_{Q'}$ or earlier, which means no fitting family of assignments for $Q$ can contain a job ending at $t_Q$.  Therefore, no job in $J$ is assigned in $D'$.  \par 
Because no job in $J$ can be assigned in $D'$, we can make a fitting family of assignments $D$ for $Q$ by adding every job in $J$ to $D'$.  This family of assignments will have profit:
\begin{equation}
\mathcal{P}(D) = \mathcal{P}(D') + \sum_{j\in J}p_j
\end{equation}
Since $A$ can formed by adding every job in $J$ to $A'$, this means:
\begin{equation}
\mathcal{P}(A) = \mathcal{P}(A') + \sum_{j\in J}p_j
\end{equation}
Now, because we know that $\mathcal{P}(D') > \mathcal{P}(A')$, we have the following inequalities:
\begin{equation}
\mathcal{P}(D') > \mathcal{P}(A')
\end{equation}
\begin{equation}
\mathcal{P}(D') + \sum_{j\in J}p_j > \mathcal{P}(A') + \sum_{j\in J}p_j
\end{equation}
\begin{equation}
\mathcal{P}(D) > \mathcal{P}(A)
\end{equation}
This means there exists a fitting family of assignments to $Q$ with profit greater than that of $A$, which means $A$ is not a MPFFA for $Q$.  This is a contradiction, as it is an assumption of the lemma that $A$ is a MPFFA for $Q$.  Therefore, $A'$ must be an MPFFA for $Q'$.
\end{proof}

What this lemma means is that removing the latest jobs from the top of an MPFFA for any subproblem yields the optimal solution to another subproblem.  That means that, for some subproblem $Q$, if we have solved every subproblem $Q'$ which could potentially have a fitting family of assignments by removing the top of $Q$, and try adding every possible set of jobs with endtime equal to the height of $Q$ into a fitting family of assignments for $Q'$, then at least one of those must be a MPFFA for $Q$. We now use the lemma to prove our main result and describe a polynomial-time exact algorithm for unit-cost MDSP.\par 

\begin{theorem}
If all jobs have a cost equal to 1 and we assume that $m$ is constant, then we can solve MDSP in polynomial time.
\end{theorem}

\begin{proof}
Therefore, we can find the MPFFA for each subproblem with the following process:
\begin{itemize}
    \item Let $Q = [t_{1,Q}... t_{m,Q},\mathcal{W}(S_{1,Q})...\mathcal{W}(S_{m,Q})]$ be a subproblem with top $T_Q$.
    \item Let \textbf{Q} be the set of every subproblem $Q' = [t_{1,Q'}... t_{m,Q'},\mathcal{W}(S_{1,Q'})...\mathcal{W}(S_{m,Q'})]$ for which the following is true: for each $i\in T_Q$, we have that $t_{i,Q'} < t_{i,Q}$ and $\mathcal{W}(S_{i,Q'}) = \mathcal{W}(S_{i,Q}) - 1$ (i.e. that the cost of drone $i$ on $Q'$ is on less than the cost of drone $i$ on $Q$, which means it takes one less job), and for each $i\not\in T_Q$, we have $t_{i,Q'} = t_{i,Q}$ and $\mathcal{W}(S_{i,Q'}) = \mathcal{W}(S_{i,Q})$.  Because every $\mathcal{W}(S_{i,Q'})$ and $t_{i,Q'}$ is less than or equal to $\mathcal{W}(S_{i,Q})$ and $t_{i,Q}$ respectively, we can assume that while finding the MPFFA of $Q$, the MPFFA of every $Q'\in \textbf{Q}$ has already been found.
    \item Let \textbf{S} be the set of every (unordered) subset $S$ of the set of all jobs for which the following is true: $|S| = |T_Q|$, and for every job $j\in S$, we have that $j$ ends at time $t_Q$.
    \item For every $Q'\in\textbf{Q}$, let $A_{Q'}$ be an MPFFA for $Q'$.  For every $S\in\textbf{S}$, attempt to add to $A_{Q'}$ in the following way: for each job $j\in S$, attempt to assign $j$ to one of the drones $i\in T_Q$.  If it is not possible to add every job, move on to the next set $S$.  If every job in $S$ can be added, then because $|S| = |T_Q|$ it will result in a fitting family of assignments for $Q$.  Record that fitting family of assignments in a set \textbf{F}.
    \item By the lemma, one of the fitting family of assignments in \textbf{F} must be an MPFFA for $Q$, which means all we need to do now is find the family of assignments in \textbf{F} which yields the largest profit, and we will have an MPFFA for $Q$.
\end{itemize}
The absolutely optimal family of assignments must be a fitting family of assignments for one of the subproblems, which means that to get the maximum profit assignment of all jobs, we simply need to search every subproblem $Q = [t_1... t_m,\mathcal{W}(S_1)... \mathcal{W}(S_m)]$ for whom each $s_i \leq B$ and find the one whose MPFFA is largest.  That MPFFA will be the maximum profit assignment of all jobs.

Now, we can find the time complexity of this algorithm.  There are $O(n^{2m})$ subproblems.  Consider the case of solving a subproblem $Q = [t_{1,Q}... t_{m,Q},\mathcal{W}(S_{1,Q})...\mathcal{W}(S_{m,Q})]$ by searching through other subproblems $Q' = [t_{1,Q'}... t_{m,Q'},\mathcal{W}(S_{1,Q'})...\mathcal{W}(S_{m,Q'})]$.  Each $Q'$ has fixed values $\mathcal{W}(S_{i,Q'})$, based on $Q$ and $T_Q$, but each $t_{i,Q'}$ can take any integer value less than $t_{i,Q}$.  Therefore, each $t_{i,Q'}$ take $O(n)$ different values, so there are $O(n^m)$ different subproblems that need to be looped through.\par 
Also, for each subproblem we encounter, we need to loop through all possible subsets of size $|T_Q|$ of every job that ends at time $t_Q$.  Because $|T_Q| \leq m$, and there are $n$ total jobs, we get that there are $O(n^m)$ subsets of every job of size $|T_Q|$.  Therefore, looping through each possible subset at each valid subproblem $Q'$ takes $O(n^m)\times O(n^m) = O(n^{2m})$ time.  Because there are $O(n^{2m})$ subproblems, this means that the total amount of time this entire algorithm takes is $O(n^{4m})$.  Because $m$ is constant, that means this algorithm is polynomial in $n$.
\end{proof}

\section{Conclusion}
\label{sec:5}
In our work, we have presented an $O(n^2)$-runtime $\frac{1}{3}$-approximation algorithm for the MDSP problem, subject to the constraint that the number of conflicting jobs at any given time does not exceed the number of drones. We have also introduced the Unit-Cost MDSP problem variant and described a dynamic programming algorithm that can discover exact solutions in polynomial time, assuming the number of drones is constant.

The algorithms we described in this paper for MDSP are all bounded approximations and thus cannot be used to get arbitrarily close to the optimum. \cite{jana22} have previously demonstrated that a Fully Polynomial-Time Approximation Scheme (FPTAS) for MDSP is not possible, since an NP-complete problem called the partition problem can be reduced to the multiple knapsack problem (MKP), which can be reduced to MDSP. They showed that if an FPTAS for MKP existed, it would be able to solve the partition problem in polynomial time, which would mean that $P = NP$.  Therefore, assuming that $P\neq NP$, there cannot exist an FPTAS for MDSP, since any approximation algorithm used on MDSP can be reduced to an algorithm that can be used for MKP.

However, the lack of an FPTAS does not imply that there does not exist a polynomial time approximation scheme (PTAS) - an algorithm that generates a $1-\epsilon$ approximation of the optimum which is polynomial in the input size but not necessarily in $1/\epsilon$.  In fact, \cite{cherukuri00} devised a PTAS for MKP, which relied on guessing polynomially many solutions which can be $1-\epsilon$ times the optimum, and then using a PTAS for bin packing to find a way to pack at least $1-\epsilon$ times the optimum in $m$ bins from one of the potentially optimal solutions. This cannot, however, directly be used as a PTAS for MDSP as MDSP is a stricter problem than MKP. Investigation into a potential PTAS for MDSP, hence, is a promising area of future research.

\section*{Acknowledgements}

This project was developed as a part of the 6.5210 (Advanced Algorithms) class taught by Prof. David Karger, of the MIT Computer Science and Artificial Intelligence Laboratory (CSAIL), at MIT in the Fall 2022 semester. We would like to express our gratitude for his feedback and support throughout the class. We would also like to thank our TAs, Theia Henderson and Michael Joseph Coulombe, for their support during the class. Finally, we would also like to thank our classmates for the feedback provided on the initial versions of this project through the peer-review assignments in the class. 

\bibliography{ref}

\begin{thebibliography}{}

\bibitem[Betti~Sorbelli et~al., 2022]{sorbelli22}
Betti~Sorbelli, F., Cor\`{o}, F., Das, S.~K., Palazzetti, L., and Pinotti,
  C.~M. (2022).
\newblock Greedy algorithms for scheduling package delivery with multiple
  drones.
\newblock In {\em 23rd International Conference on Distributed Computing and
  Networking}, ICDCN 2022, page 31–39, New York, NY, USA. Association for
  Computing Machinery.

\bibitem[Chekuri and Khanna, 2000]{cherukuri00}
Chekuri, C. and Khanna, S. (2000).
\newblock A ptas for the multiple knapsack problem.
\newblock In {\em Proceedings of the Eleventh Annual ACM-SIAM Symposium on
  Discrete Algorithms}, SODA '00, page 213–222, USA. Society for Industrial
  and Applied Mathematics.

\bibitem[Jana and Mandal, 2022]{jana22}
Jana, S. and Mandal, P.~S. (2022).
\newblock Approximation algorithms for drone delivery scheduling problem.
\newblock {\em arXiv preprint arXiv:2211.06636}.

\end{thebibliography}

\end{document}